\documentclass[envcountsect,envcountsame]{llncs}
\usepackage{amsmath}  
\usepackage{amssymb}
\usepackage{multicol}
\usepackage{multirow}
\usepackage{subfigure}
\usepackage{tikz}
\usepackage{xspace}
\usepackage{colonequals}

\usetikzlibrary{trees}

\def\savesymb#1{%
  \expandafter\let\expandafter\oldsym\expandafter=\csname#1\endcsname
  \expandafter\global\expandafter\let\csname#1BAK\endcsname=\oldsym
}
\def\restoresymb#1{%
  \expandafter\let\expandafter\oldsymb\expandafter=\csname#1BAK\endcsname
  \expandafter\global\expandafter\let\csname#1\endcsname=\oldsymb
}
\savesymb{ast}
\usepackage{mathabx}
\restoresymb{ast}

\def\qed {{                
   \parfillskip=0pt        
   \widowpenalty=10000     
   \displaywidowpenalty=10000  
   \finalhyphendemerits=0  
   \leavevmode             
   \unskip                 
   \nobreak                
   \hfil                   
   \penalty50              
   \hskip.2em              
   \null                   
   \hfill                  
   $\square$
   \par}}                  
\def\wc{\ensuremath{\mathord{\star}}\xspace}
\def\lca{\ensuremath{\mathrm{lca}}\xspace}
\def\std{\ensuremath{\mathrm{std}}\xspace}
\def\inj{\ensuremath{\mathrm{inj}}\xspace}
\def\anc{\ensuremath{\mathrm{anc}}\xspace}
\def\lab{\ensuremath{\mathit{lab}}\xspace}
\def\child{\ensuremath{\mathit{child}}\xspace}
\def\desc{\ensuremath{\mathit{desc}}\xspace}
\def\root{\ensuremath{\mathit{root}}\xspace}
\def\true{\ensuremath{\mathit{true}}\xspace}
\def\false{\ensuremath{\mathit{false}}\xspace}
\def\height{\ensuremath{\mathit{height}}\xspace}
\def\deg{\ensuremath{\mathit{deg}}\xspace}

\begin{document}

\title{On Injective Embeddings of Tree Patterns}

\author{
  Jakub Michaliszyn\inst{1} \and
  Anca Muscholl\inst{2} \and
  S\l{}awek Staworko\inst{3} \and
  Piotr Wieczorek\inst{1} \and
  Zhilin Wu\inst{4}
}

\institute{
  \fontsize{8}{10}
  \selectfont
  University of Wroc\l{}aw
  \and
  LaBRI, University of Bordeaux
  \and
  Mostrare, INRIA Lille, University of Lille
  \and
  State Key Laboratory of Computer Science, 
  Institute of Software,\\
  Chinese Academy of Sciences
}

\thispagestyle{empty}
\maketitle
\vspace{-10pt}
\begin{abstract}
  We study three different kinds of embeddings of tree patterns:
  weakly-injective, ancestor-preserving, and lca-preserving. While
  each of them is often referred to as injective embedding, they form
  a proper hierarchy and their computational properties vary (from P
  to NP-complete). We present a thorough study of the complexity of
  the model checking problem, i.e., is there an embedding of a given
  tree pattern in a given tree, and we investigate the impact of
  various restrictions imposed on the tree pattern: bound on the
  degree of a node, bound on the height, and type of allowed labels
  and edges.
\end{abstract}

\section{Introduction}
An embedding is a fundamental notion with numerous applications in
computer science, e.g., in graph pattern matching
(cf.~\cite{Fan12}). Usually, an embedding is defined as a
structure-preserving mapping that is typically required to be
injective. Tree patterns are a special class of graph patterns that
found applications, for instance in XML databases~\cite{KiMa94,ACLS02}
where they form a functional equivalent of (acyclic) conjunctive
queries for relational databases. Tree patterns are typically matched
against trees and are allowed to use special descendant edges (double
lines in Fig.~\ref{fig:embeddings}) that can be mapped to paths rather
than to single edges as it is the case with the standard child edges.

Traditionally, the semantics of tree patterns for XML is defined using
non-injective embeddings~\cite{ACLS02,MiSu04}
(Fig.~\ref{fig:std}), which is reminiscent of relational data. Since
XML data has more structure, it makes sense to exploit the tree
structure when defining tree pattern embeddings. In this context, it
is interesting to consider injective
embeddings~\cite{David08,HaDe06,FLMWW10,KiMa94}. However, the use of
descendant edges makes it cumbersome to define what exactly an
injective embedding of a tree pattern should be, and consequently,
different notions have been employed.

A \emph{weakly-injective} embedding requires only the mapping to be
injective and recent developments in graph matching suggest that such
embeddings are crucial for expressing important patterns occurring in
real life databases~\cite{FLMWW10}. They are a natural choice when we
do not wish to constrain in any way the vertical relationship of the
images of two children of some node connected with descendant
edges. However, descendant edges can be mapped to paths that
interleave, which means that even if there is a weakly-injective
embedding between a tree pattern and a tree, there need not be a
structural similarity between the tree and the tree pattern
(Fig.~\ref{fig:inj}). This is contrary to the structure-preservation
nature of embeddings and hence the prefix \emph{weakly}. One could
strengthen the restriction and prevent the embedding from introducing
vertical relationships between the nodes, which gives us
\emph{ancestor-preserving} embeddings~\cite{David08}. In this case two
descendant edges are mapped into paths that might overlap at the
beginning but eventually branch (Fig.~\ref{fig:anc}). Finally, we can
go one step further and require the paths not to overlap at all, which
translates to \emph{lca-preserving} embeddings~\cite{HaDe06}, i.e.,
embeddings that preserve lowest common ancestors of any pair of nodes
(Fig.~\ref{fig:lca}).

Unfortunately, there is a lack of a systematic and thorough treatment
of injective embeddings and there is a tendency to name each of the
embeddings above as simply injective, which could be potentially
confusing and error-prone. This paper fills this gap and shows that
injective embeddings form a proper hierarchy and that their
computational properties vary significantly (from P to NP-complete).
This further strengthens our belief that the different injective
embeddings should not be confused. More precisely, we study the
complexity of the model checking problem, i.e., given a tree pattern
$p$ and a tree $t$ is there an embedding (of a given type) of $p$ in
$t$, and we investigate the impact of various restrictions imposed on
the tree pattern: bound on the degree of a node, bound on the height,
and type of allowed labels and edges.

Our results show that while lca-preserving embeddings are in P, both
weakly-injective and ancestor-preserving embeddings are
NP-complete. Bounding the height of the pattern practically does not
change the picture but bounding the degree of a node in the pattern
renders ancestor-preserving embeddings tractable while weakly-injective
embeddings remain NP-complete. Our results show that the high
complexity springs from the use of descendant edges: if we disallow
them, the hierarchy collapses and all injective embeddings fall into
P. On the other hand, the use of node label is not essential, the
complexity remains unchanged even if we consider tree patterns using
the wildcard symbol only, essentially patterns that query only
structural properties of the tree.

Injective embeddings of tree patterns are closely related to a number
of well-established and studied notions, including \emph{tree
  inclusion}~\cite{KiMa95,Valiente05}, \emph{minor
  containment}~\cite{RoSe86,RoSe95}, \emph{subgraph
  homeomorphism}~\cite{Chung87,MaTh92}, and \emph{graph pattern
  matching}~\cite{FLMWW10}. Not surprisingly, some of our results
are subsumed by or can be easily obtained from existing results, and
conversely, there are some that are subsumed by ours (see
Sec.~\ref{sec:related-work} for a complete discussion of related
work). The principal aim of this paper is, however, to catalog the
different kinds of injective embeddings of tree patterns and identify
what aspects of tree patterns lead to intractability. To that end, all
our reductions and algorithms are new and the reductions clearly
illustrate the source of complexity of injective tree patterns.

This paper is organized as follows. In Sec.~\ref{sec:preliminaries} we
define basic notions and in Sec.~\ref{sec:injective-embeddings} we
define formally the three types of injective embeddings of tree
patterns. In Sec.~\ref{sec:complexity} we study the model checking
problem of the injective embeddings. Discussion of related work is in
Sec.~\ref{sec:related-work} and in Sec.~\ref{sec:concl-future-work} we
summarize our results and outline further directions of study. Some
proofs have been moved to appendix.

\section{Preliminaries}
\label{sec:preliminaries}
We assume a fixed and finite set of node labels $\Sigma$ and use a
wildcard symbol $\wc$ not present in $\Sigma$. A \emph{tree
  pattern}~\cite{KiMa94,ACLS02} is a tuple
$p=(N_p,\root_p,\lab_p,\child_p,\desc_p)$, where $N_p$ is a finite set
of nodes, $\root_p\in N_p$ is the root node,
$\lab_p:N_p\rightarrow\Sigma\cup\{\wc\}$ is a labeling function,
$\child_p\subseteq N_p\times N_p$ is a set of child edges, and
$\desc_p\subseteq N_p\times N_p$ is a set of (proper) descendant
edges. We assume that $\child_p\cap\desc_p=\emptyset$, that the
relation $\child_p\cup\desc_p$ is acyclic and require every non-root
node to have exactly one predecessor in this relation. A \emph{tree}
is a tree pattern that has no descendant edges and uses no wildcard
symbols $\wc$.

An example of a tree pattern can be found in Fig.~\ref{fig:embeddings}
(descendant edges are drawn with double lines). Sometimes, we use
unranked terms to represent trees and the standard XPath syntax to
represent tree patterns. XPath allows to navigate the tree with a
syntax similar to directory paths used in the UNIX file system. For
instance, in Fig.~\ref{fig:embeddings} $p_0$ can be written as
$f/a[.//b/c]//b$. In the sequel, we use $p,p_0,p_1,\ldots$ to range
over tree patterns and $t,t_0,t_1,\dots$ to range over trees.


Given a binary relation $R$, we denote by $R^+$ the transitive closure
of $R$, and by $R^\ast$ the transitive and reflexive closure of
$R$. Now, fix a pattern $p$ and take two of its nodes $n,n'\in
N_p$. We say that $n'$ is a \emph{$|$-child} of $n$ if
$(n,n')\in\child_p$, $n'$ is a \emph{$\|$-child} of $n$ if
$(n,n')\in\desc_p$, and $n'$ is simply a \emph{child} of $n$ in $p$ if
$(n,n')\in\child_p\cup\desc_p$. Also, $n'$ is a \emph{descendant} of
$n$, and $n$ an \emph{ancestor} of $n'$, if
$(n,n')\in(\child_p\cup\desc_p)^\ast$. Note that descendantship and
ancestorship are reflexive: a node is its own ancestor and its own
descendant. The \emph{depth} of a node $n$ in $p$ is the length of the
path from the root node $\root_p$ to $n$, and here, a path is a
sequence of edges, and in particular, the depth of the root node is
$0$.  The \emph{lowest common ancestor} of $n$ and $n'$ in $p$,
denoted by $\lca_p(n,n')$, is the deepest node that is an ancestor of
$n$ and $n'$. The \emph{size} of a tree pattern $p$, denoted $|p|$, is
the number of its nodes. The \emph{degree} of a node $n$, denoted
$\deg_p(n)$, is the number of its children. The \emph{height} of a
tree pattern $p$, denoted $\height(p)$, is the depth of its deepest
node.

The standard semantics of tree patterns is defined using non-injective
embeddings which map the nodes of a tree pattern to the nodes of a
tree in a manner that respects the wildcard and the semantics of the
edges. Formally, an {\em embedding} of a tree pattern $p$ in a tree
$t$ is a function $h : N_p \rightarrow N_t$ such that:
\begin{enumerate}
\itemsep0pt
\item[$1$.] $h(\root_p)=\root_t$,
\item[$2$.] for every $(n,n')\in\child_p$,
  $(h(n),h(n'))\in\child_t$,
\item[$3$.] for every $(n,n')\in\desc_p$,
  $(h(n),h(n'))\in(\child_t)^+$,
\item[$4$.] for every $n\in N_p$, $\lab_t(h(n))=\lab_p(n)$ unless
  $\lab_p(n)=\wc$.
\end{enumerate}
We write $t \preccurlyeq_\std p$ if there exists a (standard)
embedding of $p$ in $t$. Note that the semantics of a descendant edge
of the tree pattern is in fact that of a \emph{proper descendant}: a
descendant edge is mapped to a nonempty path in the tree.

\section{Injective embeddings}
\label{sec:injective-embeddings}
We identify three subclasses of injective embeddings that restrict the
standard embedding by adding one additional condition each. First, we
have the \emph{weakly-injective} embedding of $p$ in $t$
($t\preccurlyeq_\inj p$):
\begin{enumerate}
\item[$5'$.] $h$ is an injective function, i.e., $h(n_1)\neq h(n_2)$
  for any two different nodes $n_1$ and $n_2$ of $p$.
\end{enumerate}
Next, we have the \emph{ancestor-preserving} embedding of $p$ in $t$
($t\preccurlyeq_\anc p$):
\begin{enumerate}
\item[$5''$.] $h(n_1)$ is an ancestor of $h(n_2)$ in $t$ if and only
  if $n_1$ is an ancestor of $n_2$ in $p$, for any two nodes $n_1$ and
  $n_2$ of $p$. More formally, for any $n_1,n_2\in N_p$
  \[
  (h(n_1),h(n_2))\in\child_t^\ast \iff
  (n_1,n_2)\in(\child_p\cup\desc_p)^\ast.
  \]
\end{enumerate}
Finally, we have the \emph{lca-preserving} embedding of $p$ in $t$
($t\preccurlyeq_\lca p$):
\begin{enumerate}
\item[$5'''$.] $h$ maps the lowest common ancestor of nodes $n_1$ and
  $n_2$ to the lowest common ancestor of $h(n_1)$ and $h(n_2)$, i.e.,
  for any pair of nodes $n_1$ and $n_2$ of $p$ we have
  $\lca_t(h(n_1),h(n_2)) =  h(\lca_p(n_1,n_2))$.
\end{enumerate}
In Fig.~\ref{fig:embeddings} we illustrate various embeddings of a
tree pattern $p_0$. 
\begin{figure}[htb]
  \centering
  \subfigure[non-injective $t_0\preccurlyeq_\std p_0$\label{fig:std}]{
  \begin{tikzpicture}[yscale=0.6]
    \path[use as bounding box] (-3.75,1) rectangle (1.75,-4.5);
    \begin{scope}[xshift=-2.5cm]
    \node at (0,0)      (m0) {$f$};
    \node at (0,-1)     (m1) {$a$} edge[-,semithick] (m0);
    \node at (0,-2)     (m3) {$b$} edge[-,semithick] (m1);
    \node at (0,-3)     (m4) {$c$} edge[-,semithick] (m3);
    \end{scope}
    \node at (0,0)      (n0) {$f$};
    \node at (0,-1)     (n1) {$a$} edge[-,semithick] (n0);
    \node at (0.75,-2)  (n2) {$b$} edge[-,double,semithick] (n1);
    \node at (-0.75,-2) (n3) {$b$} edge[-,double,semithick] (n1);
    \node at (-0.75,-3) (n4) {$c$} edge[-,semithick] (n3);
    \begin{scope}[semithick,blue!60!black,>=stealth]
    \draw (n0) edge[->,bend right] (m0);
    \draw (n1) edge[->,bend right] (m1);
    \draw (n2) edge[->,bend right] (m3);
    \draw (n3) edge[->,bend left] (m3);
    \draw (n4) edge[->,bend left] (m4);
    \end{scope}
  \end{tikzpicture}
  }
  \subfigure[weakly-injective $t_1\preccurlyeq_\inj p_0$\label{fig:inj}]{
  \begin{tikzpicture}[yscale=0.6,yshift=0.5cm]
    \path[use as bounding box] (-3.75,1) rectangle (1.75,-4.5);
    \begin{scope}[xshift=-2.5cm]
    \node at (0,0)      (m0) {$f$};
    \node at (0,-1)     (m1) {$a$} edge[-,semithick] (m0);
    \node at (0,-2)     (m2) {$b$} edge[-,semithick] (m1);
    \node at (0,-3)     (m3) {$b$} edge[-,semithick] (m2);
    \node at (0,-4)     (m4) {$c$} edge[-,semithick] (m3);
    \end{scope}
    \node at (0,0)      (n0) {$f$};
    \node at (0,-1)     (n1) {$a$} edge[-,semithick] (n0);
    \node at (0.75,-2)  (n2) {$b$} edge[-,double,semithick] (n1);
    \node at (-0.75,-2) (n3) {$b$} edge[-,double,semithick] (n1);
    \node at (-0.75,-3) (n4) {$c$} edge[-,semithick] (n3);
    \begin{scope}[semithick,blue!60!black,>=stealth]
    \draw (n0) edge[->,bend right] (m0);
    \draw (n1) edge[->,bend right] (m1);
    \draw (n2) edge[->,bend right] (m2);
    \draw[bend angle=20] (n3) edge[->,bend right] (m3);
    \draw[bend angle=20] (n4) edge[->,bend right] (m4);
    \end{scope}
  \end{tikzpicture}
  }
  \subfigure[ancestor-preserving $t_2\preccurlyeq_\anc p_0$\label{fig:anc}]{
  \begin{tikzpicture}[yscale=0.6]
    \path[use as bounding box] (-3.75,1) rectangle (1.75,-4.5);
    \begin{scope}[xshift=-2.5cm]
    \node at (0,0)      (m0) {$f$};
    \node at (0,-1)     (m1) {$a$} edge[-,semithick] (m0);
    \node at (0,-2)     (mo) {$g$} edge[-,semithick] (m1);
    \node at (0.75,-3)  (m2) {$b$} edge[-,semithick] (mo);
    \node at (-0.75,-3) (m3) {$b$} edge[-,semithick] (mo);
    \node at (-0.75,-4) (m4) {$c$} edge[-,semithick] (m3);
    \end{scope}
    \node at (0,0)      (n0) {$f$};
    \node at (0,-1)     (n1) {$a$} edge[-,semithick] (n0);
    \node at (0.75,-2)  (n2) {$b$} edge[-,double,semithick] (n1);
    \node at (-0.75,-2) (n3) {$b$} edge[-,double,semithick] (n1);
    \node at (-0.75,-3) (n4) {$c$} edge[-,semithick] (n3);
    \begin{scope}[semithick,blue!60!black,>=stealth]
    \draw (n0) edge[->,bend right] (m0);
    \draw (n1) edge[->,bend right] (m1);
    \draw[bend angle=15] (n2) edge[->,bend left] (m2);
    \draw[bend angle=15] (n3) edge[->,bend right] (m3);
    \draw[bend angle=15] (n4) edge[->,bend left] (m4);
    \end{scope}
  \end{tikzpicture}
  }
  \subfigure[lca-preserving $t_3\preccurlyeq_\lca p_0$\label{fig:lca}]{
  \begin{tikzpicture}[yscale=0.6]
    \path[use as bounding box] (-3.75,1) rectangle (1.75,-4.5);
    \begin{scope}[xshift=-2.5cm]
    \node at (0,0)      (m0) {$f$};
    \node at (0,-1)     (m1) {$a$} edge[-,semithick] (m0);
    \node at (0.75,-2)  (m2) {$b$} edge[-,semithick] (m1);
    \node at (-0.75,-2) (mo) {$g$} edge[-,semithick] (m1);
    \node at (-0.75,-3) (m3) {$b$} edge[-,semithick] (mo);
    \node at (-0.75,-4) (m4) {$c$} edge[-,semithick] (m3);
    \end{scope}
    \node at (0,0)      (n0) {$f$};
    \node at (0,-1)     (n1) {$a$} edge[-,semithick] (n0);
    \node at (0.75,-2)  (n2) {$b$} edge[-,double,semithick] (n1);
    \node at (-0.75,-2) (n3) {$b$} edge[-,double,semithick] (n1);
    \node at (-0.75,-3) (n4) {$c$} edge[-,semithick] (n3);
    \begin{scope}[semithick,blue!60!black,>=stealth]
    \draw (n0) edge[->,bend right] (m0);
    \draw (n1) edge[->,bend right] (m1);
    \draw (n2) edge[->,bend right] (m2);
    \draw[bend angle=15] (n3) edge[->,bend left] (m3);
    \draw[bend angle=15] (n4) edge[->,bend left] (m4);
    \end{scope}
  \end{tikzpicture}
  }
  \caption{Embeddings of a tree pattern $p_0$.}
  \label{fig:embeddings}
\end{figure}
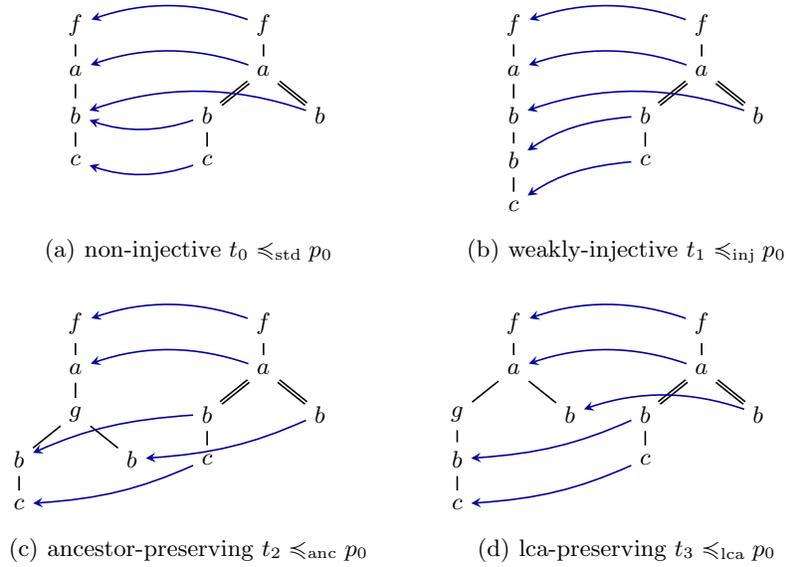

We point out that injective embeddings form a hierarchy, and in
particular, lca-preserving and ancestor-preserving embeddings are
weakly-injective.
\begin{proposition}
  \label{prop:1}
  For any tree $t$ and tree pattern $p$, 1)
  $t\preccurlyeq_\lca p \Rightarrow t \preccurlyeq_\anc p$, 2)
  $t\preccurlyeq_\anc p \Rightarrow t \preccurlyeq_\inj p$, and 3)
  $t\preccurlyeq_\inj p \Rightarrow t \preccurlyeq_\std p$.
\end{proposition}
It is also easy to see that the hierarchy is proper. For that, take
Fig.~\ref{fig:embeddings} and note that $t_0\preccurlyeq_\std p_0$ but
$t_0\not\preccurlyeq_\inj p_0$, $t_1\preccurlyeq_\inj p_0$ but
$t_1\not\preccurlyeq_\anc p_0$, and finally, $t_2\preccurlyeq_\anc
p_0$ but $t_2\not\preccurlyeq_\lca p_0$. We point out, however, that
the hierarchy of injective embeddings collapses if we disallow
descendant edges in tree patterns.
\begin{proposition}
  \label{prop:4}
  For any tree $t$ and any tree pattern $p$ that does not use
  descendant edges, $t\preccurlyeq_\inj p$ iff $t\preccurlyeq_\anc p$
  iff $t\preccurlyeq_\lca p$.
\end{proposition}
Furthermore, if we consider path patterns, i.e., tree patterns whose
nodes have at most one child, there is no difference between any of
the injective embeddings and the standard embedding.

\begin{proposition}
\label{prop:6}
For any tree $t$ and any path pattern $p$, $t\preccurlyeq_\std p$ iff
$t\preccurlyeq_\inj p$ iff $t\preccurlyeq_\anc p$ iff
$t\preccurlyeq_\lca p$.
\end{proposition}
\section{Complexity of injective embeddings}
\label{sec:complexity}
For a type of embedding $\theta\in\{\inj,\anc,\lca\}$ we define the
corresponding (unconstrained) decision problem:
\[
\mathcal{M}_\theta=\{(t,p) \mid t\preccurlyeq_\theta p\}.
\]
Additionally, we investigate several constrained variants of this
problem. First, we restrict the degree of nodes in the tree pattern
by a constant $k \geq 0$,
\[
\mathcal{M}_\theta^{\mathrm{D} \leq k}=\{(t,p) \mid t
\preccurlyeq_\theta p,\ \forall n\in N_p.\ \deg_p(n) \leq k
\}.
\]
Next, we define the restriction of the height of the tree pattern by a constant
$k \geq 0$,
\[
\mathcal{M}_\theta^{\mathrm{H} \leq k}=\{(t,p) \mid t
\preccurlyeq_\theta p,\ \height(p) \leq k \}.
\]
We also investigate the importance of labels in tree patterns as
opposed to those that are label-oblivious and query only the structure
of the tree, i.e., tree patterns that use $\wc$ only.
\[
\mathcal{M}_\theta^{\wc} =\{ (t,p) \mid t \preccurlyeq_\theta p,\
\forall n\in N_p.\ \lab_p(n)=\wc \}.
\]
It is also interesting to see if disallowing $\wc$ may change the
picture.
\[
\mathcal{M}_\theta^{\circ} =\{ (t,p) \mid t \preccurlyeq_\theta p,\
\forall n\in N_p.\ \lab_p(n)\neq\wc \}.
\]
Finally, we restrict the use of child and descendant edges in the tree
pattern. 
\[
\mathcal{M}_\theta^{|} =
\{ (t,p) \mid t \preccurlyeq_\theta p,\
\desc_p=\emptyset \}\quad\text{and}\quad
\mathcal{M}_\theta^{\|} =
  \{ (t,p) \mid t \preccurlyeq_\theta p,\
  \child_p=\emptyset \}.
\]
We make several general observations. First, we point out that the
conditions on the various injective embeddings can be easily verified
and every embedding is a mapping whose size is bounded by the size of
the tree pattern. Therefore,
\begin{proposition}
  \label{prop:2}
  $\mathcal{M}_\theta$, $\mathcal{M}_\theta^{\mathrm{D}\leq k}$,
  $\mathcal{M}_\theta^{\mathrm{H}\leq k}$, $\mathcal{M}_\theta^{\wc}$,
  $\mathcal{M}_\theta^{\circ}$, $\mathcal{M}_\theta^{|}$, and
  $\mathcal{M}_\theta^{\|}$ are in NP for any
  $\theta\in\{\inj,\anc,\lca\}$ and $k\geq 0$.
\end{proposition}
By Prop.~\ref{prop:6}, for path patterns we employ the existing
polynomial algorithm~\cite{GoKoPi03}.
\begin{proposition}
  \label{prop:7}
  $\mathcal{M}_\theta^{D\leq1}$ is in P for any
  $\theta\in\{\inj,\anc,\lca\}$.
\end{proposition}
Finally, by Prop.~\ref{prop:4} and Thm.~\ref{thm:5}, which shows the
tractability of lca-preserving embeddings, we get the following.
\begin{proposition}
  \label{prop:8}
  $\mathcal{M}_\theta^{|}$ is in P for any
  $\theta\in\{\inj,\anc,\lca\}$.
\end{proposition}
\subsection{Weakly-injective embeddings}
\begin{theorem}
  \label{thm:6}
  $\mathcal{M}_\inj$ is NP-complete.
\end{theorem}
\begin{proof}
  We reduce SAT to $\mathcal{M}_\inj$. We take a CNF formula
  $\varphi=c_1\land\dots\land c_k$ over the variables $x_1,\ldots,x_n$
  and for every variable $x_i$ we construct two (linear) trees
  $X_i=x_i(\pi_1(\pi_2(\ldots\pi_{k-1}(\pi_k)\ldots)))$ and
  $\bar{X}_i=x_i(\bar\pi_1(\bar\pi_2
  (\ldots\bar\pi_{k-1}(\bar\pi_k)\ldots)))$, where $\pi_j=c_j$ if the
  clause $c_j$ uses the literal $x_i$ and $\pi_j=\bot$ otherwise, and
  analogously, $\bar\pi_j=c_j$ if the clause $c_j$ uses the literal
  $\lnot x_i$ and $\bar\pi_j=\bot$ otherwise. The constructed tree is
  \[
  t_\varphi = r(X_1,\bar{X}_1,X_2,\bar{X}_2,\ldots, X_n,\bar{X}_n)
  \]
  and the constructed tree pattern is 
  \[
  p_\varphi = r[.//Y_1][.//Y_2]\dots[.//Y_n]
  [.//c_1][.//c_2]\ldots[.//c_k],
  \]
  where $Y_i = x_i/\wc/\wc/\ldots/\wc$ with exactly $k$ repetitions of
  $\wc$.   
  \begin{figure}[htb]
    \centering
    \begin{tikzpicture}[semithick,scale=0.75]
      \node at (-2,0) {$p_\varphi$};
      \node at (0, 0) (n0)  {$r$};
      \node at (0.5,-1) (n1) {$c_1$} edge[-,double] (n0);
      \node at (1.5,-1) (n2) {$c_2$} edge[-,double] (n0);
      \node at (2.5,-1) (n3) {$c_3$} edge[-,double] (n0);
      \node at (-2.5,-1) (n41) {$x_1$} edge[-,double] (n0);
      \node at (-2.5,-2) (n42) {$\wc$} edge[-] (n41);
      \node at (-2.5,-3) (n43) {$\wc$} edge[-] (n42);
      \node at (-2.5,-4) (n44) {$\wc$} edge[-] (n43);
      \node at (-1.5,-1) (n51) {$x_2$} edge[-,double] (n0);
      \node at (-1.5,-2) (n52) {$\wc$} edge[-] (n51);
      \node at (-1.5,-3) (n53) {$\wc$} edge[-] (n52);
      \node at (-1.5,-4) (n54) {$\wc$} edge[-] (n53);
      \node at (-0.5,-1) (n61) {$x_3$} edge[-,double] (n0);
      \node at (-0.5,-2) (n62) {$\wc$} edge[-] (n61);
      \node at (-0.5,-3) (n63) {$\wc$} edge[-] (n62);
      \node at (-0.5,-4) (n64) {$\wc$} edge[-] (n63);

      \begin{scope}[xshift=-6.5cm]
      \node at (-2,0) {$t_\varphi$};
      \node at (0, 0) (n0)  {$r$};
      \node at (-2.5,-1) (n41) {$x_1$} edge[-] (n0);
      \node at (-2.5,-2) (n42) {$c_1$} edge[-] (n41);
      \node at (-2.5,-3) (n43) {$c_2$} edge[-] (n42);
      \node at (-2.5,-4) (n44) {$\bot$}edge[-] (n43);
      \node at (-1.5,-1) (n51) {$x_1$} edge[-] (n0);
      \node at (-1.5,-2) (n52) {$\bot$}edge[-] (n51);
      \node at (-1.5,-3) (n53) {$\bot$}edge[-] (n52);
      \node at (-1.5,-4) (n54) {$c_3$} edge[-] (n53);      
      \node at (-0.5,-1) (n61) {$x_2$} edge[-] (n0);
      \node at (-0.5,-2) (n62) {$\bot$}edge[-] (n61);
      \node at (-0.5,-3) (n63) {$\bot$}edge[-] (n62);
      \node at (-0.5,-4) (n64) {$\bot$}edge[-] (n63);
      \node at (+0.5,-1) (n71) {$x_2$} edge[-] (n0);
      \node at (+0.5,-2) (n72) {$\bot$}edge[-] (n71);
      \node at (+0.5,-3) (n73) {$c_2$} edge[-] (n72);
      \node at (+0.5,-4) (n74) {$c_3$} edge[-] (n73);
      \node at (+1.5,-1) (n81) {$x_3$} edge[-] (n0);
      \node at (+1.5,-2) (n82) {$\bot$}edge[-] (n81);
      \node at (+1.5,-3) (n83) {$c_2$} edge[-] (n82);
      \node at (+1.5,-4) (n84) {$\bot$}edge[-] (n83);
      \node at (+2.5,-1) (n91) {$x_3$} edge[-] (n0);
      \node at (+2.5,-2) (n92) {$c_1$}edge[-] (n91);
      \node at (+2.5,-3) (n93) {$\bot$}edge[-] (n92);
      \node at (+2.5,-4) (n94) {$\bot$}edge[-] (n93);
      \end{scope}
    \end{tikzpicture}
    \caption{Reduction to $\mathcal{M}_\inj$ for
      $\varphi=(x_1\lor \lnot x_3) \land
      (x_1\lor\lnot x_2\lor x_3) \land (\neg x_1 \lor \lnot x_2)$.}
    \label{fig:reduction-inj}
  \end{figure}
  Figure~\ref{fig:reduction-inj} illustrates the reduction for
  $\varphi=(x_1\lor \lnot x_3) \land (x_1\lor\lnot x_2\lor x_3) \land
  (\neg x_1 \lor \lnot x_2)$. We claim that
  \[
  (t_\varphi,p_\varphi) \in \mathcal{M}_\inj \iff  \varphi \in \mathrm{SAT}.
  \]
  For the \emph{if} part, we take a valuation $V$ satisfying $\varphi$
  and construct a weakly-injective embedding $h$ as follows. The fragment
  $[.//Y_i]$ is mapped to $\bar{X}_i$ if $V(x_i)=\true$ and to $X_i$ if
  $V(x_i)=\false$. For each clause $c_j$ we pick one literal satisfied
  by $V$ and w.l.o.g.\ assume it is $x_i$, i.e., $c_j$ uses $x_i$ and
  $V(x_i)=\true$. Then, the embedding $h$ maps the fragment $[.//c_j]$
  to the node $c_j$ in the tree fragment $X_i$. Clearly, the
  constructed embedding is an injective function.

  For the \emph{only if} part, we take a weakly-injective embedding
  $h$ and construct a satisfying valuation $V$ as follows. If the
  fragment $[.//Y_i]$ is mapped to $X_i$, then $V(x_i)=\false$ and if
  $[.//Y_i]$ is mapped to $\bar{X}_i$, then $V(x_i)=\true$. To show
  that $\varphi$ is satisfied by $V$ we take any clause $c_j$ and
  check where $h$ maps the fragment $[.//c_j]$. W.l.o.g.\ assume that
  it is $X_i$ and since $h$ is weakly-injective, $Y_i$ is mapped to
  $\bar{X}_i$, and consequently, $V(x_i)=\true$. Hence, $V$ satisfies
  $c_j$.  \qed
\end{proof}
We observe that in the reduction above the use of the child edges in
the tree pattern is not essential and they can be replaced by
descendant edges.
\begin{corollary}
  \label{cor:3}
  $\mathcal{M}_{\inj}^{\|}$ is NP-complete.
\end{corollary}
Furthermore, the proof of Thm.~\ref{thm:6} can be easily adapted to
the bounded degree setting. Indeed, one can easily show that for any
tree $t=r(t_1,\ldots,t_k)$ and any tree pattern
$p=r[.//p_1]\dots[.//p_m]$, $t\preccurlyeq_\inj p$ if and only if
$t'\preccurlyeq_\inj p'$, where
$t'=A_1(\ldots{}A_m(t_1,\ldots,t_k)\dots)$,
$p'=A_1[.//p_1]/\ldots/A_m[.//p_m]$, and $A_1,\ldots,A_m$ are new
symbols not used in $p$. This observation, when applied to the tree
pattern in the reduction above, allows to reduce the degree of the
root node and to obtain a tree pattern of degree bounded by $2$. Note,
however, that this technique does not allow to reduce the degree of
nodes in arbitrary tree patterns.
\begin{corollary}
  \label{cor:1}
  $\mathcal{M}_\inj^{\mathrm{D}\leq k}$ is NP-complete for any
  $k\geq2$.
\end{corollary}
A reduction similar to the one presented above can be used to
construct patterns whose height is exactly $2$.
\begin{theorem}
  \label{thm:4}
  $\mathcal{M}_\inj^{\mathrm{H}\leq k}$ is NP-complete for any $k\geq
  2$.
\end{theorem}
If we consider patterns of depth $1$, where the children of the root
node are leaves, then a diligent counting technique suffices to solve
the problem.
\begin{proposition}
  \label{prop:3}
  $\mathcal{M}_\inj^{\mathrm{H}\leq 1}$ is in P.
\end{proposition}
\begin{proof}
  Fix a tree pattern $p$ whose depth is $1$ and a tree $t$. For
  $a\in\Sigma\cup\{\wc\}$ we denote by $\smash{p_a^|}$ the number of
  $a$-labeled $|$-children of $\root_p$, by $\smash{p_a^{\|}}$ the
  number of $a$-labeled $\|$-children of $\root_p$, and by
  $\smash{t_a^{{}=i}}$ and $\smash{t_a^{{}\geq i}}$ the numbers of
  $a$-labeled nodes of $t$ at depths ${}=i$ and ${}\geq i$ resp. 

  We attempt to construct a weakly-injective embedding of $p$ to $t$
  using the following strategy: (1) we map the nodes of
  $\smash{p_a^|}$ to nodes of $\smash{t_a^{{}=1}}$, (2) we map the
  nodes of $\smash{p_a^{\|}}$ to nodes of $\smash{t_a^{{}\geq2}}$ and
  if $\smash{p_a^{\|} > t_a^{{}\geq2}}$, we map the remaining
  $\smash{p_a^{\|}-t_a^{{}\geq2}}$ nodes to the nodes of
  $\smash{t_a^{{}=1}}$, (3) we map the nodes of $\smash{p_{\wc}^|}$ to
  the remaining nodes of $t$ at depth $1$, and (4) we map the nodes of
  $\smash{p_{\wc}^{\|}}$ to the remaining nodes of $t$.

  Clearly, this procedure succeeds and a weakly-injective embedding
  can be constructed if and only if the following inequalities are
  satisfied:
  \begin{align}
    &\smash{p_a^| \leq t_a^{{}=1}}&
    &\text{for $a\in\Sigma$,}\\
    &\smash{p_a^{\|} \leq t_a^{{}\geq1} - p_a^{|}}&
    &\text{for $a\in\Sigma$,}\\
    &\smash{p_{\wc}^| \leq \textstyle{\sum_{a\in\Sigma}}
      (t_a^{{}=1} - p_a^{|} - \min(p_a^{\|}-t_a^{{}\geq2},0)}),\\
    &\smash{p_{\wc}^{\|} \leq 
      \Big[\textstyle{\sum_{a\in\Sigma}}
        (t_a^{{}\geq1} - p_a^{|} - p_a^{\|})\Big]
      -p_{\wc}^{|}}.
  \end{align}
  Naturally, these inequalities can be verified in polynomial
  time. \qed
\end{proof}
Finally, we observe that while in the reductions above we use
different labels to represent elements of a finite enumerable set, the
same can be accomplished with patterns using $\wc$ labels only, where
natural numbers are encoded with simple gadgets. The gadgets use the
fact that a node of a tree pattern that has $k$ $|$-children can be
mapped by a weakly-injective embedding only to a node having at least
$k$ nodes. On the other hand, we can easily modify reduction from
Thm.~\ref{thm:6} yield tree patterns without $\wc$ nodes. 
\begin{theorem}
  \label{thm:7}
  $\mathcal{M}_\inj^{\wc}$ and $\mathcal{M}_\inj^\circ$ are  NP-complete.
\end{theorem}
\subsection{Ancestor-preserving embeddings}
\begin{theorem}
  \label{thm:1}
  $\mathcal{M}_\anc$ is NP-complete.
\end{theorem}
\begin{proof}
  To prove NP-hardness we reduce SAT to $\mathcal{M}_\anc$. We take a
  formula $\varphi=c_1\land c_2 \land \ldots \land c_k$ over variables
  $x_1,\ldots,x_n$ and for every variable $x_i$ we construct two trees:
  $X_i=x_i(c_{j_1},\ldots,c_{j_m})$ such that $c_{j_1},\ldots,c_{j_m}$
  are exactly the clauses satisfied by using the literal $x_i$, and
  $\bar{X}_i=x_i(c_{j_1},\ldots,c_{j_m})$ such that
  $c_{j_1},\ldots,c_{j_m}$ are exactly the clauses using the literal
  $\neg x_i$. The constructed tree is
  \[
  t_\varphi = r(X_1,\bar{X}_1,\ldots,X_n,\bar{X}_n).
  \]
  And the tree pattern (written in XPath syntax) is 
  \[
  p_\varphi = r[x_1]\ldots[x_n][.//c_1]\ldots[.//c_k].
  \]
  An example of the reduction for $\varphi=(x_1\lor \lnot x_3) \land
  (x_1\lor\lnot x_2\lor x_3) \land (\neg x_1 \lor \lnot x_2)$ is
  presented in Fig.~\ref{fig:reduction-anc}.
  \begin{figure}[htb]
  \centering
  \begin{tikzpicture}[semithick,yscale=0.75]
    \begin{scope}[grow via three points={%
        one child at (0,-1) and two children at (-.25,-1) and (.25,-1)}]
    \node at (-1.5,0) {$t_\varphi$};
    \node (r) at (0,0) {$r$};
    \begin{scope}[yshift=-1cm, xshift=-2.5cm]
    \node at (0,0) {$x_1$} edge[-] (r)
      child {node {$c_1$}}
      child {node {$c_2$}}
      ;
    \end{scope}
    \begin{scope}[yshift=-1cm, xshift=-1.25cm]
    \node at (0,0) {$x_1$} edge[-] (r)
      child {node {$c_3$}}
      ;
    \end{scope}
    \begin{scope}[yshift=-1cm, xshift=-.5cm]
    \node at (0,0) {$x_2$} edge[-] (r)
      ;
    \end{scope}
    \begin{scope}[yshift=-1cm, xshift=0.5cm]
    \node at (0,0) {$x_2$} edge[-] (r)
      child {node {$c_2$}}
      child {node {$c_3$}}
      ;
    \end{scope}
    \begin{scope}[yshift=-1cm, xshift=1.75cm]
    \node at (0,0) {$x_3$} edge[-] (r)
      child {node {$c_2$}}
      ;
    \end{scope}
    \begin{scope}[yshift=-1cm, xshift=2.75cm]
    \node at (0,0) {$x_3$} edge[-] (r)
      child {node {$c_1$}}
      ;
    \end{scope}
    \end{scope}
    \begin{scope}[xshift=5.25cm]
      \node at (-1,0) {$p_\varphi$};
      \node (r) at (0,0) {$r$};
      \node (X) at (-1.25,-1) {$x_1$} edge[-] (r);
      \node (X) at (-0.75,-1) {$x_2$} edge[-] (r);
      \node (X) at (-0.25,-1) {$x_3$} edge[-] (r);
      \node (c) at (0.25,-1) {$c_1$} edge[-,double] (r);
      \node (c) at (0.75,-1) {$c_2$} edge[-,double] (r);
      \node (c) at (1.25,-1) {$c_3$} edge[-,double] (r);
    \end{scope}
  \end{tikzpicture}
  \caption{Reduction to $\mathcal{M}_\anc$ for $\varphi=(x_1\lor \lnot
    x_3) \land (x_1\lor\lnot x_2\lor x_3) \land (\neg x_1 \lor \lnot
    x_2)$.}
  \label{fig:reduction-anc}
  \end{figure}
  The main claim is that $(t_\varphi,p_\varphi)\in\mathcal{M}_\anc$
  iff $\varphi\in\mathrm{SAT}$. We prove it analogously to the main
  claim in the proof of Theorem~\ref{thm:6}. The use of
  ancestor-preserving embeddings ensures that the fragments $[x_i]$
  and $[.//c_j]$ are not mapped to the same subtree of $t_\varphi$,
  and this reduction does not work for weakly-injective embeddings.
  \qed
\end{proof}
We point out that in the proof above, the constructed pattern has
height $1$.
\begin{corollary}
  \label{cor:2}
  $\mathcal{M}_\anc^{\mathrm{H}\leq k}$ is NP-complete for every
  $k\geq 1$.
\end{corollary}
Also, the use of child edges is not essential and they can be replaced
by descendant edges and the reduction does not use $\wc$ labels.
\begin{corollary}
  \label{cor:4}
  $\mathcal{M}_\anc^{\|}$ and $\mathcal{M}_\anc^{\circ}$ are
  NP-complete.
\end{corollary}
Bounding the degree of a node in the tree pattern renders, however,
checking the existence of an ancestor-preserving embedding tractable.
\begin{theorem}
  \label{thm:2}
  For any $k\geq 0$, $\mathcal{M}_\anc^{\mathrm{D}\leq k}$ is in P.
\end{theorem}
\begin{proof}
  We fix a tree $t$ and a tree pattern $p$. For a node $m\in N_p$ we
  define $\Phi(m) = \{n\in N_t \mid t|_n \preccurlyeq_\anc p|_m\}$,
  where $t|_n$ is a subtree of $t$ rooted at $n$ (and similarly, we
  define $p|_m$). Naturally, $t\preccurlyeq_\anc p$ iff
  $\root_t\in\Phi(\root_p)$.

  We fix a node $m\in N_p$ with children $m_1,\ldots,m_k$, suppose
  that we have computed $\Phi(m_i)$ for every $i\in\{1,\ldots,k\}$,
  and take a node $n\in N_t$. We claim that $n$ belongs to $\Phi(m)$
  if and only if the following two conditions are satisfied: 1)
  $\lab_t(n) = \lab_p(m)$ unless $\lab_p(m)=\wc$, 2) there is
  $(n_1,\ldots,n_k)\in\Phi(m_1)\times\ldots\times\Phi(m_k)$ such that
  a) $n_i$ is not an ancestor of $n_j$ for all $i\neq j$, b)
  $(n,n_i)\in\child_t$ if $(m,m_i)\in\child_p$, and c)
  $(n,n_i)\in\child_t^+$ if $(m,m_i)\in\desc_p$.

  Since $k$ is bounded by a constant, the product
  $\Phi(m_1)\times\ldots\times\Phi(m_k)$ is of size polynomial in the
  size of $t$, and therefore, the whole procedure works in polynomial
  time too. \qed
\end{proof}
Finally, gadgets similar to those in Thm.~\ref{thm:7} allow us dispose
of labels altogether.
\begin{theorem}
  \label{thm:8}
  $\mathcal{M}_\anc^{\wc}$ is NP-complete.
\end{theorem}

\subsection{LCA-preserving embeddings}
\begin{theorem}
  \label{thm:5}
  $\mathcal{M}_\lca$ is in P.
\end{theorem}
\begin{proof}
  We fix a tree $t$ and a tree pattern $p$. For a node $m\in N_p$ we
  define $ \Phi(m) = \{n\in N_t \mid t|_n \preccurlyeq_\lca p|_m\}, $
  where $t|_n$ is a subtree of $t$ rooted at $n$ (and similarly, we
  define $p|_m$).  Naturally, $t \preccurlyeq_\lca p$ if and only if
  $\root_t\in \Phi(\root_p)$. We present a bottom-up procedure for
  computing $\Phi$.
  
  We fix a node $m\in N_p$ with children $m_1,\ldots,m_k$, suppose
  that we have computed $\Phi(m_i)$ for every $i\in\{1,\ldots,k\}$,
  take a node $n\in N_t$, and let $n_1,\ldots,n_\ell$ be its
  children. We claim that $n$ belongs to $\Phi(m)$ if and only if the
  following two conditions are satisfied: 1) $\lab_t(n) = \lab_p(m)$
  unless $\lab_p(m)=\wc$ and 2) the bipartite graph $G=(X\cup Y,E)$
  with $X=\{m_1,\ldots,m_k\}$, $Y=\{n_1,\ldots,n_\ell\}$, and 
  \begin{align*}
    &E=\{
    (m_i,n_j) 
    \mid 
    (m,m_i)\in\child_p \land n_j\in\Phi(m_i)\lor{}\\
    &\phantom{E=\{(m_i,n_j)\mid{}}
    (m,m_i)\in\desc_p \land \exists n'\in\Phi(m_i).\ 
    (n_j,n')\in\child_t^\ast.
    \},
  \end{align*}
  has a matching of size $k$. In the construction of $E$ we use the
  expression $(n_j,n')\in\child_t^\ast$ because a $\|$-child $m_i$ of
  $m$ needs to be connected with proper descendants of $n$ and these
  are descendants of $n_j$'s. We finish by pointing out that a maximum
  matching of $G$ can be constructed in polynomial
  time~\cite{HoKa73}. \qed
\end{proof}
\section{Related work}
\label{sec:related-work}
Model checking for tree patterns has been studied in the literature in
a variety of variants depending on the requirements on the
corresponding embeddings. They may, or may not, have to be injective,
preserve various properties like the order among siblings, ancestor or
child relationships, label equalities, etc.  In this paper, we
consider unordered, injective embeddings that additionally may be
ancestor- or lca-preserving.

Kilpel{\"a}inen and Mannila \cite{KiMa95} studied the
\emph{unordered tree inclusion} problem defined as
follows. Given labeled trees $P$ and $T$, can $P$ be obtained from $T$
by deleting nodes? Here, deleting a node $u$ entails removing all
edges incident to $u$ and, if $u$ has a parent $v$, replacing the edge
from $v$ to $u$ by edges from $v$ to the children of $u$. The
unordered tree inclusion problem is equivalent to the model checking
for ancestor-preserving embeddings where the tree pattern contains
descendants edges only. \cite{KiMa95} shows NP-completeness for tree
patterns of height 1. Moreover, \cite{MaTh92} shows that the problem
remains NP-complete when all labels in both trees are $\wc$ or
when degrees of all vertices except root are at most 3. These two
results subsume our
Thm.~\ref{thm:1}~and~\ref{thm:8}. \cite{KiMa95,MaTh92} show also the
tractability of the problem when the degrees of all nodes in the tree
pattern are bounded. Thm. \ref{thm:2} generalizes this to allow also
for child edges in the tree patterns.

The tree inclusion problem is a special case of the \emph{minor
  containment} problem for graphs \cite{RoSe86,RoSe95}: given two
graphs $G$ and $H$, decide whether $G$ contains $H$ as a minor, or
equivalently, whether $H$ can be obtained from a subgraph of $G$ by
edge contractions, where contracting an edge means replacing the edge
and two incident vertices by a single new vertex. For trees, edge
contraction is equivalent to node deletion.  Since minor containment
is known to be NP-complete, even for trees, this gives another proof
of NP-completeness for ancestor-preserving embeddings.

Valiente \cite{Valiente05} introduced the \emph{constrained unordered
  tree inclusion} problem where the question is, given labeled trees
$P$ and $T$, whether $P$ can be obtained from $T$ by deleting nodes of
degrees one or two. The polynomial time algorithm given there is based
on the earlier results on subtree homeomorphisms \cite{Chung87} where
unlabeled trees are considered.  The constrained unordered tree
inclusion is equivalent to model checking of lca-preserving embeddings
where all edges in the tree pattern are descendants.  Our
Thm.~\ref{thm:5} slightly generalizes the above result allowing also
for child edges in the pattern.

David~\cite{David08} studied the complexity of ancestor-preserving
embeddings of tree patterns with data comparison (equality and
inequality) and showed their NP-completeness. Although we show that
ancestor-preserving embeddings are NP-complete even without data
comparisons, the reductions used in~\cite{David08} construct tree
patterns of a bounded degree, which shows that adding data comparisons
indeed increases the computational complexity of the model checking
problem.

Recently, Fan et al.~\cite{FLMWW10} studied 1-1 $p$-homomorphisms
which extend injective graph homomorphisms by relaxing the edge
preservation condition. Namely, the edges have to be mapped to
nonempty paths. However, neither the internal vertices nor edges
within the paths have to be disjoint.
In case of trees, 1-1 p-homomorphisms correspond to the
weakly-injective embeddings that we consider in this paper.  By
reduction from \emph{exact cover by 3-sets problem} they have shown
NP-completeness of model checking in the case where the first graph is
a tree and the second is a DAG. We improve this result in
Thm.~\ref{thm:6}~and~\ref{thm:7}.

When embeddings have to preserve the order among siblings, model
checking becomes much easier. The \emph{ordered tree inclusion
  problem} was initially introduced by Knuth \cite[exercise
2.3.2-22]{Knuth68} who gave a sufficient condition for testing
inclusion. The polynomial time algorithms from \cite{KiMa95} is based
on dynamic programming and at each level may compute the inclusion
greedily from left-to-right thanks to the order preservation
requirement.  The tree inclusion is also related to the \emph{ordered
  tree pattern matching} \cite{HoDo82}, where embeddings have to
preserve the order and child-relationship, but they do not have
necessarily to preserve root.

\section{Conclusions and future work}
\label{sec:concl-future-work}
We have considered three different notions of injective embeddings of
tree patterns and for each of them we have studied the problem of
model checking. 
\begin{table}[htb]
  \vspace{-10pt}
  \fontsize{8}{10}
  \selectfont
  \centering
  \begin{tabular}{|cc|c|c|c|c|}
    \cline{3-6}
    \multicolumn{2}{c|}{}&
    $~\preccurlyeq_\std~$ & 
    $\preccurlyeq_\inj$ & 
    $\preccurlyeq_\anc$ &
    $\preccurlyeq_\lca$\\
    \hline
    \multicolumn{2}{|c|}{unconstrained} 
    &
    \multirow{9}{*}{P~\cite{GoKoPi03}}
    &
    NP-c.~(Thm.~\ref{thm:6})
    &
    NP-c.~(Thm.~\ref{thm:1})~\cite{KiMa95,MaTh92}
    &
    \multirow{9}{*}{
      \begin{tabular}[c]{c}
        P\\(Thm.~\ref{thm:5})
      \end{tabular}
    }
    \\
    \cline{1-2}
    \cline{4-5}
    \multirow{2}{*}{$k$-bounded degree}
    &
    $k\geq 2$ 
    &
    &
    NP-c.~(Cor.~\ref{cor:1})
    &
    \multirow{2}{*}{P~(Thm.~\ref{thm:2})}
    &
    \\
    \cline{2-2}
    \cline{4-4}
    &
    $k=1$&
    &
    P~(Prop.~\ref{prop:7})
    &
    &
    \\
    \cline{1-2}
    \cline{4-5}
    \multirow{2}{*}{$k$-bounded height}&
    $k\geq 2$&
    &
    NP-c.~(Thm.~\ref{thm:7})
    &
    \multirow{2}{*}{NP-c.~(Thm.~\ref{thm:1})~\cite{KiMa95}}
    &
    \\
    \cline{2-2}
    \cline{4-4}
    &
    $k=1$&
    &
    P~(Prop.~\ref{prop:3})
    &
    &
    \\
    \cline{1-2}
    \cline{4-5}
    \multicolumn{2}{|c|}{$\wc$ labels only}&
    &
    NP-c.~(Thm.~\ref{thm:7})
    &
    NP-c.~(Thm.~\ref{thm:8})~\cite{MaTh92}
    &
    \\
    \cline{1-2}
    \cline{4-5}
    \multicolumn{2}{|c|}{no $\wc$ labels}&
    &
    NP-c.~(Thm.~\ref{thm:7})
    &
    NP-c.~(Cor.~\ref{cor:4})~\cite{KiMa95}
    &
    \\
    \cline{1-2}
    \cline{4-5}
    \multicolumn{2}{|c|}{no $|$-edges}&
    &
    NP-c.~(Cor.~\ref{cor:3})
    &
    NP-c.~(Cor.~\ref{cor:4})~\cite{KiMa95,MaTh92}
    &
    \\
    \cline{1-2}
    \cline{4-5}
    \multicolumn{2}{|c|}{no $\|$-edges}&
    &
    \multicolumn{2}{c|}{P~(Prop.~\ref{prop:8})}&
    \\
    \hline
  \end{tabular}
  \vspace{5pt}
  \normalsize
  \caption{Summary of complexity results}
  \label{tab:summary}
\end{table}
Table~\ref{tab:summary} summarizes the complexity results. All our
results extend to embeddings between pairs of tree patterns, used for
instance in static query analysis~\cite{MiSu04}. Although some of our
results are subsumed by or can be easily obtained from existing
results, our reductions and algorithms are simple and clean.  In
particular, we show intractability with direct reductions from SAT.

In the future, we would like to find out whether there is an algorithm
for checking lca-preserving embeddings that does not rely on
constructing perfect matchings in bipartite graphs. The exact bound on
complexity of non-injective embeddings of tree patters is a difficult
open problem~\cite{GoKoMa09} and it would be interesting if
establishing exact bounds on tractable cases of injective embeddings
is any easier.


\clearpage
\appendix
\section*{Appendix: Omitted proofs}

\noindent 
\textbf{Proposition~\ref{prop:1}.}
\begin{it}
For any tree $t$ and tree pattern $p$, 1) $t\preccurlyeq_\lca p
\Rightarrow t \preccurlyeq_\anc p$, 2) $t\preccurlyeq_\anc p
\Rightarrow t \preccurlyeq_\inj p$, and 3) $t\preccurlyeq_\inj p
\Rightarrow t \preccurlyeq_\std p$.
\end{it}
\noindent
\begin{proof}
Assume that $t\preccurlyeq_\lca p$ and let $h$ be a $\lca$-preserving embedding. Consider any $n_1, n_2$ such that $(h(n_1), h(n_2)) \in child^\ast_t$. Since $h$ is $\lca$-preserving and $lca(h(n_1), h(n_2))=h(n_1)$, $lca(n_1, n_2) = n_1$ and therefore $n_1$ is an ancestor of $n_2$. So $h$ is ancestor-preserving.

For the proof of 2, consider $p, t$ such that $t\preccurlyeq_\anc p$ and let $h$ be an ancestor-preserving embedding. We show that $h$ is injective. Assume that there are $n_1, n_2$ such that $h(n_1)=h(n_2)$. Since $h$ is ancestor-preserving, $(h(n_1), h(n_2)) \in \child_t^\ast$, and $(h(n_2), h(n_1)) \in \child_t^\ast$, $(n_1,n_2)\in(\child_p\cup\desc_p)^\ast$ and $(n_2,n_1)\in(\child_p\cup\desc_p)^\ast$, so $n_1=n_2$.
$t \preccurlyeq_\inj p$.

Finally, Implication 3 follows from the fact that any injective embedding is an embedding.
\qed
\end{proof}

\noindent 
\textbf{Proposition~\ref{prop:4}.}
\begin{it}
  For any tree $t$ and any tree pattern $p$ that does not use
  descendant edges, $t\preccurlyeq_\inj p$ iff $t\preccurlyeq_\anc p$
  iff $t\preccurlyeq_\lca p$.
\end{it}
\begin{proof}
Assume that $p$ does not use descendant edges. 
By Proposition \ref{prop:1}, it is enough to prove that  $t\preccurlyeq_\inj p$ implies $t\preccurlyeq_\lca p$.

Let $t\preccurlyeq_\inj p$ and $h$ be an embedding from $p$ to $t$. It is easy to see that $h$ is an isomorphisms on a substructure of $t$, and therefore $h$ is \lca{}-preserving and $t\preccurlyeq_\lca p$. Together with Proposition \ref{prop:1} it implies all the equivalences.
\qed
\end{proof}

\noindent 
\textbf{Proposition~\ref{prop:6}.}
\begin{it}
  For any tree $t$ and any path pattern $p$, $t\preccurlyeq_\std$ iff
  $t\preccurlyeq_\inj p$ iff $t\preccurlyeq_\anc p$ iff
  $t\preccurlyeq_\lca p$.
\end{it}
\begin{proof}
Assume that $p$ is a path pattern and $t\preccurlyeq_\std p$ and let $h$ be an embedding from $p$ to $t$. Consider any nodes $n, m$ of $p$ such that there is a path from $n$ to $m$. Clearly, $lca(n, m)=n$. By Properties 2 and 3 of the definition of embeddings, there is also a path from $h(n)$ to $h(m)$, hence $lca(h(n), h(m))=h(n)$. Therefore $h$ is \lca{}-preserving and $t\preccurlyeq_\lca p$. By Proposition \ref{prop:1} we obtain the required equivalence.
\qed
\end{proof}

\noindent
\textbf{Theorem~\ref{thm:4}.}
\begin{it}
  $\mathcal{M}_\inj^{\mathrm{H}\leq k}$ is NP-complete for any $k\geq
  2$.
\end{it}
\begin{proof}
We show how to build, for a given instance $\varphi$ of SAT problem, a pattern $p_\varphi$ and a tree $t_\varphi$ such that $t_\varphi\preccurlyeq_\inj p_\varphi$ if and only if $\varphi$ is satisfiable. 
Let $\varphi = c_1 \wedge c_2 \wedge \dots \wedge c_k$ be an instance of SAT over variables $x_1, \dots, x_n$. We set $\Sigma=\{a, c_1,$ \dots$, c_k, s_1, $\dots$, s_n\}$.

For each $i$, we define the tree $X_i$ as follows. Its root is a node $x_i^p$ and it is connected to $k+1$ nodes, namely $x_i^n$, $p_i^1$, $p_i^2$, \dots, $p_i^k$. Node $x_i^n$ has $k+2$ successors, namely $s_i$, $n_i^1, \dots, n_i^{k+1}$. All other nodes have no successors.  

The tree $t_\varphi$ consists of a root $r$ and its $n$ disjoint successors --- $X_1, \dots, X_n$ (see Fig. \ref{fig:reduction-inj-stp}).

Now we define the labeling of $t_\varphi$. Let $c_{i_1}, \dots, c_{i_l}$ be the clauses with the positive occurrence of $x_i$, and $c_{j_1}, \dots, c_{j_{l'}}$ be the clauses with the negative occurrence of $x_i$. For all $s\leq l$, we label $p_s$ with $c_{i_s}$. Similarly, for all $s\leq l'$ we label $n_s$ by $c_{j_s}$. We label $s_i$ by $s_i$ and all other nodes by $a$.

The pattern $p_\varphi$ is as presented at Fig. \ref{fig:reduction-inj-stp}. Clearly, its depth is bounded by $2$.

  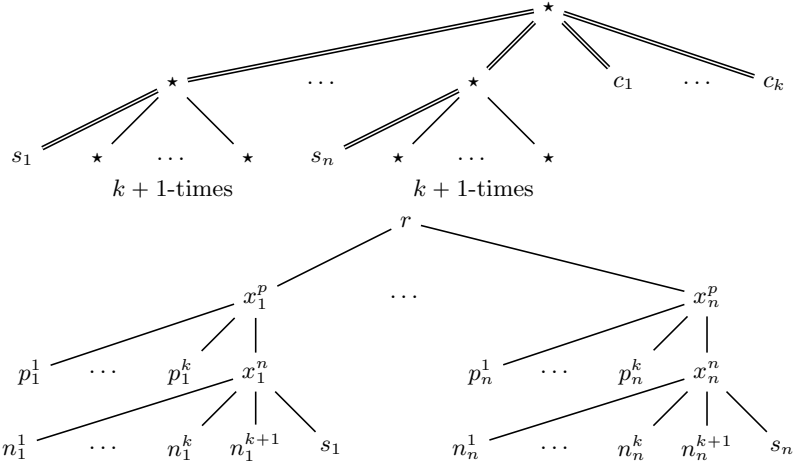
\begin{figure}[ht!]
    \centering
    \begin{tikzpicture}[semithick]
      \node at (0, 0) (n0)  {$\wc$};
      \node at (-5,-1) (n1)  {$\wc$}   edge[-,double] (n0);
      \node at (-7,-2) (n11) {$s_1$} edge[-,double] (n1);
      \node at (-6,-2) (n11) {$\wc$} edge[-] (n1);
      \node at (-4,-2) (n11) {$\wc$} edge[-] (n1);
      \node at (-1,-1) (nn)  {$\wc$}   edge[-,double] (n0);
      \node at (-3,-2) (nn1) {$s_n$} edge[-,double] (nn);
      \node at (-2,-2) (nn1) {$\wc$} edge[-] (nn);
      \node at (0,-2) (nn1) {$\wc$} edge[-] (nn);
      \node at (1,-1) (c1)  {$c_1$}   edge[-,double] (n0);
      \node at (3,-1) (ck)  {$c_k$}   edge[-,double] (n0);
\coordinate [label=center:{\dots}] (A) at (-1, -2);
\coordinate [label=center:{$k+1$-times}] (A) at (-1, -2.4);
\coordinate [label=center:{\dots}] (A) at (-5, -2);
\coordinate [label=center:{$k+1$-times}] (A) at (-5, -2.4);
\coordinate [label=center:{\dots}] (A) at (2, -1);
\coordinate [label=center:{\dots}] (A) at (-3, -1);
    \end{tikzpicture}
\\

    \begin{tikzpicture}[semithick]
      \node at (5, 0) (n0)  {$r$};
      \node at (3,-1) (n1)  {$x_1^p$}   edge[-] (n0);
      \node at (0,-2) (n1s)  {$p_1^1$}   edge[-] (n1);
      \node at (2,-2) (n1s)  {$p_1^k$}   edge[-] (n1);
      \node at (3,-2) (n11)  {$x_1^n$}   edge[-] (n1);
      \node at (-0.2,-3) (n11s)  {$n_1^1$}   edge[-] (n11);
      \node at (2,-3) (n11s)  {$n_1^{k}$}   edge[-] (n11);
      \node at (3,-3) (n111) {$n_1^{k+1}$} edge[-] (n11);
      \node at (4,-3) (n111) {$s_1$} edge[-] (n11);
      \node at (9,-1) (nn)  {$x_n^p$}   edge[-] (n0);
      \node at (6,-2) (n1s)  {$p_n^1$}   edge[-] (nn);
      \node at (8,-2) (n1s)  {$p_n^k$}   edge[-] (nn);
      \node at (9,-2) (nn1)  {$x_n^n$}   edge[-] (nn);
      \node at (5.8,-3) (n11s)  {$n_n^1$}   edge[-] (nn1);
      \node at (8,-3) (n11s)  {$n_n^{k}$}   edge[-] (nn1);
      \node at (9,-3) (n11s)  {$n_n^{k+1}$}   edge[-] (nn1);
      \node at (10,-3) (nn11) {$s_n$} edge[-] (nn1);
\coordinate [label=center:{\dots}] (A) at (5, -1);
\coordinate [label=center:{\dots}] (A) at (1, -2);
\coordinate [label=center:{\dots}] (A) at (7, -2);
\coordinate [label=center:{\dots}] (A) at (1, -3);
\coordinate [label=center:{\dots}] (A) at (7, -3);
    \end{tikzpicture}
    \caption{The pattern $p_\varphi$ (at the top) and the tree $t_\varphi$ (without the labeling).}
    \label{fig:reduction-inj-stp}
  \end{figure}

Assume that $t_\varphi\preccurlyeq_\inj p_\varphi$ and let $h$ be the corresponding embedding. Let $Y$ be the set of all the successors of $root_p$ labeled by $\wc$ in $p_\varphi$ and $h(Y)$ be the image of $Y$. A quick check shows that for each $i$ there is exactly one node from $\{x_i^p,x_i^n\}$ in $h(Y)$. We define the valuation for $\varphi$ such that $x_i$ is positive if $x_i^n \in h(Y)$ and negative otherwise. 

Consider any clause $c_s$ and let $m$ be the node in $p_\varphi$ labeled by $c_s$. Assume that $h(m)=p_i^j$ for some $i, j$. It means that $x_i$ occurs positively in $c_s$ and that $x_i^p$ does not belong to $h(Y)$ --- otherwise, if $x_i^p = h(m')$ for some $m'$, then all successors of $x_i^p$ would be results of $h$ applied to successors of $m'$, contradicting the facts that $h(m)=p_i^j$ and $h$ is injective. Therefore, $c_s$ is satisfied.

The proof that if $\varphi$ is satisfiable then $t_\varphi\preccurlyeq_\inj p_\varphi$ should now be straightforward. 
\qed
\end{proof}

\noindent
\textbf{Theorem~\ref{thm:7}.}
\begin{it}
  $\mathcal{M}_\inj^{\wc}$ and $\mathcal{M}_\inj^\circ$ are
  NP-complete.
\end{it}
\begin{proof}
  \begin{figure}[t!]
    \centering
    \begin{tikzpicture}[semithick]
      \node at (0, 0) (n0)  {$\wc$};
      \node at (0,-1) (n1)  {$\wc$}   edge[-] (n0);
      \node at (1,-1) (n11)  {$\wc$}   edge[-] (n0);
      \node at (3,-1) (n14)  {$\wc$}   edge[-] (n0);
      \node at (0,-2) (n2)  {$\wc$}   edge[-] (n1);
      \node at (0,-3) (n4)  {$\wc$}   edge[dotted] (n2);
      \node at (0,-4) (n5) {$\wc$} edge[-] (n4);
      \node at (0,-5) (n6) {$\wc$} edge[-] (n5);
      \node at (1,-5) (n61) {$\wc$} edge[-] (n5);
      \node at (3,-5) (n64) {$\wc$} edge[-] (n5);
\coordinate [label=center:{$k+3$-times}] (A) at (2, -1.4);
\coordinate [label=center:{$k+3$-times}] (A) at (2, -5.4);
\coordinate [label=center:{\dots}] (A) at (2, -1);
\coordinate [label=center:{\dots}] (A) at (2, -5);
\coordinate [label=right:{(s nodes)}] (A) at (0.15, -2.5);
    \end{tikzpicture}
    \caption{The tree $T^k_s$.}
    \label{fig:reduction-inj-ts}
  \end{figure} 
  
For the $\mathcal{M}_\inj^{\wc}$  case, we simply adjust the proof of  Theorem~\ref{thm:4}, taking the advantage of the fact that all the nodes with labels different than $\wc$ are in leaves. 

For each $s \in \mathbb{N}$, we define tree $T^k_s$ that consists of two nodes with $k+3$ successors and a path connecting them of length $s$ (see Fig. \ref{fig:reduction-inj-ts}). Note that in the original $t_\varphi$ does not contain any node of degree $\geq k+3$.

We replace all nodes (in $t_\varphi$ and $p_\varphi$) labeled by $c_s$ by $T^k_s$ and all nodes labeled by $n_s$ by $T^k_{k+s}$. Then, in $t_\varphi$, we replace all labels by $a$. It is readily checked that for any $i \neq j$ there is no embedding from $T_i^k$ to $T_j^k$, and since $t_\varphi$ contains no nodes with degree at least $k+3$, there is an embedding from the modified pattern to the modified tree if and only if $\varphi$ is satisfiable.

For the $\mathcal{M}_\inj^{\circ}$  case, we simply replace all $\wc$ in the pattern defined above by $a$, the only label present in the tree.
  \qed
\end{proof}

\noindent
\textbf{Theorem~\ref{thm:8}.}
\begin{it}
  $\mathcal{M}_\anc^{\wc}$ is  are NP-complete.
\end{it}
\begin{proof}
We modify the proof of Theorem \ref{thm:1}. First, we adjust the tree and the pattern by adding one node below each $x_i$, label it by $x_i$ and label old $x_i$ by $a$. We also replace $r$ by $a$ and all $a$ in the pattern by $\wc$ (see Fig. \ref{fig:newold}).

  \begin{figure}[htb]
  \centering
  \begin{tikzpicture}[semithick,yscale=0.75]
    \begin{scope}[grow via three points={%
        one child at (0,-1) and two children at (-.25,-1) and (.25,-1)}]
    \node at (-1.5,0) {$t_\varphi$};
    \node (r) at (0,0) {$a$};
    \begin{scope}[yshift=-1cm, xshift=-2.5cm]
    \node at (0,0) {$a$} edge[-] (r)
      child {node {$x_1$}}
      child {node {$c_1$}}
      child {node {$c_2$}}
      ;
    \end{scope}
    \begin{scope}[yshift=-1cm, xshift=-1.25cm]
    \node at (0,0) {$a$} edge[-] (r)
      child {node {$x_1$}}
      child {node {$c_3$}}
      ;
    \end{scope}
    \begin{scope}[yshift=-1cm, xshift=-.5cm]
    \node at (0,0) {$a$} edge[-] (r)
      child {node {$x_2$}}
      ;
    \end{scope}
    \begin{scope}[yshift=-1cm, xshift=0.5cm]
    \node at (0,0) {$a$} edge[-] (r)
      child {node {$x_2$}}
      child {node {$c_2$}}
      child {node {$c_3$}}
      ;
    \end{scope}
    \begin{scope}[yshift=-1cm, xshift=1.75cm]
    \node at (0,0) {$a$} edge[-] (r)
      child {node {$x_3$}}
	  child {node {$c_2$}}
      ;
    \end{scope}
    \begin{scope}[yshift=-1cm, xshift=2.75cm]
    \node at (0,0) {$a$} edge[-] (r)
      child {node {$x_3$}}
      child {node {$c_1$}}
      ;
    \end{scope}
    \end{scope}
    \begin{scope}[xshift=5.25cm]
      \node at (-1,0) {$p_\varphi$};
      \node (r) at (0,0) {$\wc$};
      \node (X1) at (-1.25,-1) {$\wc$} edge[-] (r);
      \node (X2) at (-0.75,-1) {$\wc$} edge[-] (r);
      \node (X3) at (-0.25,-1) {$\wc$} edge[-] (r);
      \node (a) at (-1.25,-2) {$x_1$} edge[-] (X1);
      \node (a) at (-0.75,-2) {$x_2$} edge[-] (X2);
      \node (a) at (-0.25,-2) {$x_3$} edge[-] (X3);
      \node (c) at (0.25,-1) {$c_1$} edge[-,double] (r);
      \node (c) at (0.75,-1) {$c_2$} edge[-,double] (r);
      \node (c) at (1.25,-1) {$c_3$} edge[-,double] (r);
    \end{scope}
  \end{tikzpicture}
  \caption{Adjusted reduction to $\mathcal{M}_\anc$ for $\varphi=(x_1\lor \lnot
    x_3) \land (x_1\lor\lnot x_2\lor x_3) \land (\neg x_1 \lor \lnot
    x_2)$.}
  \label{fig:newold}
  \end{figure}

The obtained tree and pattern have the following property: the only nodes that are not labeled by $\wc$ or $a$ are leaves.
By virtually the same way as in Theorem \ref{thm:7} we can replace them by trees $T_s$.
  \qed
\end{proof}
\end{document}